\documentclass{article}%
\usepackage{amsmath}
\usepackage{amsfonts}
\usepackage{amssymb}
\usepackage{graphicx}
\usepackage[ruled,vlined,linesnumbered]{algorithm2e}%
\providecommand{\U}[1]{\protect\rule{.1in}{.1in}}
\newtheorem{theorem}{Theorem} [section]

\newtheorem{conjecture}[theorem]{Conjecture}
\newtheorem{corollary}[theorem]{Corollary}

\newtheorem{lemma}[theorem]{Lemma}

\newenvironment{proof}[1][Proof]{\noindent\textbf{#1.} }{\ \rule{0.5em}{0.5em}}
\begin{document}

\author{David Tankus\\Department of Software Engineering\\Sami Shamoon College of Engineering, Ashdod, ISRAEL\\davidt@sce.ac.il}
\title{Recognizing Generating Subgraphs in Graphs without Cycles of Lengths $6$ and 7}
\date{}
\maketitle

\begin{abstract}
Let $B$ be an induced complete bipartite subgraph of $G$ on vertex sets of bipartition $B_{X}$ and $B_{Y}$. The subgraph $B$ is {\it generating} if there exists an independent set $S$ such that each of $S \cup B_{X}$ and $S \cup B_{Y}$ is a maximal independent set in the graph. If $B$ is generating, it \textit{produces} the restriction $w(B_{X})=w(B_{Y})$.  

Let $w:V(G) \longrightarrow\mathbb{R}$ be a weight function. We say that $G$ is $w$-well-covered if all maximal independent sets are of the same weight. The graph $G$ is $w$-well-covered if and only if $w$ satisfies all restrictions produced by all generating subgraphs of $G$. Therefore, generating subgraphs play an important role in characterizing weighted well-covered graphs.

It is an \textbf{NP}-complete problem to decide whether a subgraph is generating, even when the subgraph is isomorphic to $K_{1,1}$ \cite{bnz:related}. We present a polynomial algorithm for recognizing generating subgraphs for graphs without cycles of lengths 6 and 7.
\end{abstract}

\section{Introduction}
\subsection{Definitions and Notation}
Throughout this paper $G=(V,E)$ is a simple (i.e., a finite, undirected,
loopless and without multiple edges) graph with vertex set $V = V(G)$ and edge
set $E = E(G)$.

Cycles of $k$ vertices are denoted by $C_{k}$. When we say that $G$ does not
contain $C_{k}$ for some $k \geq3$, we mean that $G$ does not admit subgraphs
isomorphic to $C_{k}$. It is important to mention that these subgraphs are not
necessarily induced. Let $\mathcal{G}(\widehat{C_{i_{1}}},..,\widehat{C_{i_{k}%
}})$ denote the family of all graphs which do not contain $C_{i_{1}}$,...,$C_{i_{k}}$.

Let $S\subseteq V$ be a non-empty set of vertices, and let $i\in\mathbb{N}$. Then
\[
N_{i}(S)=\{v \in V| \ min_{s\in S}\ d(v,s)=i\}, \ 
N_{i}[S]=\{v \in V| \ min_{s\in S}\ d(v,s)\leq i\}
\]
where $d(x,y)$ is the minimal number of edges required to construct a path
between $x$ and $y$. If $i\neq j$ then $N_{i}(S)\cap N_{j}(S)$$=$$\emptyset$. We abbreviate $N_{1}(S)$ and $N_{1}[S]$ to $N(S)$ and $N[S]$, respectively. If
$S=\{v\}$ for some $v\in V$, then $N_{i}(\{v\})$, $N_{i}[\{v\}]$, $N(\{v\})$, and $N[\{v\}]$, are abbreviated to $N_{i}(v)$, $N_{i}[v]$, $N(v)$, and $N[v]$, respectively.

A set of vertices $S\subseteq V$ is \textit{independent }if for every $x,y\in
S$, $x$ and $y$ are not adjacent. It is clear that an empty set is
independent. An independent set is called {\it maximal} if it is not contained in another independent set. An independent set is {\it maximum} if the graph does not contain an independent set with a higher cardinality.
A graph is called \textit{well-covered} if every maximal independent set is maximum. The problem of finding a maximum cardinality
independent set in an input graph is \textbf{NPC}. However, if the input is
restricted to well-covered graphs, then a maximum cardinality independent set
can be found polynomially using the \textit{greedy algorithm}.

Let $T\subseteq V$. Then $S$ \textit{dominates} $T$ if $T \subseteq N[S]$. If $S$ and $T$ are both empty, then $N(S)=\emptyset$, and $S$
dominates $T$. If $S$ is a maximal independent set of $G$, then it dominates
the whole graph.

Let $w:V \longrightarrow\mathbb{R}$ be a weight function defined on the
vertices of $G$. For every set $S \subseteq V$, define $w(S)=\Sigma_{s \in
S}w(s)$. Then $G$ is {\it $w$-well-covered} if all maximal independent sets of $G$
are of the same weight. The set of weight functions $w$ for which $G$ is
$w$-well-covered is a \textit{vector space} \cite{cer:degree}.

The recognition of well-covered graphs is known to be \textbf{co-NP}-complete.
This was proved independently in \cite{cs:note} and \cite{sknryn:compwc}. In
\cite{cst:structures} it is proven that the problem remains \textbf{co-NP}%
-complete even when the input is restricted to $K_{1,4}$-free graphs. However,
the problem is polynomially solvable for $K_{1,3}$-free graphs
\cite{tata:wck13f,tata:wck13fn}, for graphs with girth at least $5$
\cite{fhn:wcg5}, for graphs that contain neither $4$- nor $5$-cycles
\cite{fhn:wc45}, for graphs with a bounded maximal degree \cite{cer:degree},
and for chordal graphs \cite{ptv:chordal}.

Recently, Levit and Tankus constructed a polynomial time algorithm for finding
the vector space of weight functions $w$ such that the input graph
$G\in\mathcal{G}(\widehat{C_{4}},\widehat{C_{5}},\widehat{C_{6}})$ is $w$-well-covered \cite{lt:wc456}. They used the
following notion. Let $B$ be an induced complete bipartite subgraph of $G$ on
vertex sets of bipartition $B_{X}$ and $B_{Y}$. Assume that there exists an
independent set $S$ such that each of $S\cup B_{X}$ and $S\cup B_{Y}$ is a
maximal independent set of $G$. Then $B$ is called a \textit{generating} subgraph of
$G$, and it \textit{produces} the restriction: $w(B_{X})=w(B_{Y})$. The set $S$ is called a \textit{witness} that $B$ is generating. A graph $G$ is $w$-well-covered for a weight function $w:V(G) \longrightarrow\mathbb{R}$ if and only if $w$ satisfies all restrictions produced by all generating subgraphs of $G$. Therefore, generating subgraphs play an important role in characterizing $w$-well-covered graphs.  

In the restricted case that the generating subgraph $B$ is isomorphic to
$K_{1,1}$, call its vertices $x$ and $y$. In that case $x$ and $y$ are said to
be \textit{related}, $xy$ is a \textit{relating} edge, and $w(x)=w(y)$ for
every weight function $w$ such that $G$ is $w$-well-covered. The witness of
the related vertices $x$ and $y$ is an independent set $S$, containing neither
$x$ nor $y$, such that both $S\cup\{x\}$ and $S\cup\{y\}$ are maximal
independent sets in the graph. 

The decision problem whether an edge in an input graph is relating is
\textbf{NP-}complete \cite{bnz:related}, and it remains \textbf{NP-}complete even
when the input is restricted to graphs without cycles of lengths $4$ and $5$
\cite{lt:relatedc4} or to bipartite graphs \cite{lt:npc}. Therefore,
recognizing generating subgraphs is also \textbf{NP-}complete in these cases.
However, recognizing relating edges can be done in polynomial time if the
input is restricted to graphs without cycles of lengths $4$ and $6$
\cite{lt:relatedc4}, and to graphs without cycles of lengths $5$ and $6$
\cite{lt:wc456}.

Recognizing generating subgraphs is \textbf{NP-}complete when the input is
restricted to $K_{1,4}$-free graphs \cite{lt:npc} or to graphs with girth at
least 6 \cite{lt:npc}. However, the problem is polynomial solvable when the
input is restricted to graphs without cycles of lengths $4$, $6$ and $7$
\cite{lt:wc4567}, to graphs without cycles of lengths $4$, $5$ and $6$
\cite{lt:wc456}, and to graphs without cycles of lengths $5$, $6$ and $7$
\cite{lt:wc456}.

\subsection{Main Results}
The subject of this paper is graphs without cycles of lengths 6 and 7. In Section 2 we define {\it extendable vertices}, and present a polynomial algorithm for recognizing extendable vertices in graphs without cycles of lengths 6 and 7.

\begin{theorem}
\label{c67ext} The following problem can be solved in polynomial time:\\
Input: A graph $G \in \mathcal{G}(\widehat{C}_{6},\widehat{C}_{7})$, and
a vertex $x \in V(G)$.\\
Question: Is $x$ extendable?
\end{theorem}

In Section 3 we use Theorem \ref{c67ext} to prove Theorem \ref{c67related}.

\begin{theorem}
\label{c67related} The following problem can be solved in polynomial time:\\
Input: A graph $G \in \mathcal{G}(\widehat{C}_{6},\widehat{C}_{7})$, and
an edge $xy\in E(G)$.\\
Question: Is $xy$ a relating edge?
\end{theorem}

It is proved in \cite{lt:wc4567} and \cite{lt:wc456} that recognizing generating subgraphs can be done polynomially for $\mathcal{G}(\widehat{C}_{4},\widehat{C}_{6},\widehat{C}_{7})$, and $\mathcal{G}(\widehat{C}_{5},\widehat{C}_{6},\widehat{C}_{7})$.

\begin{theorem}
\label{c467gen} \cite{lt:wc4567} The following problem can be solved in polynomial time:\\
Input: A graph $G \in \mathcal{G}(\widehat{C}_{4},\widehat{C}_{6},\widehat{C}_{7})$, and
an induced complete bipartite subgraph $B$.\\
Question: Is $B$ generating?
\end{theorem}

\begin{theorem}
\label{c567gen} \cite{lt:wc456} The following problem can be solved in polynomial time:\\
Input: A graph $G \in \mathcal{G}(\widehat{C}_{5},\widehat{C}_{6},\widehat{C}_{7})$, and
an induced complete bipartite subgraph $B$.\\
Question: Is $B$ generating?
\end{theorem}

Theorem \ref{c467gen} and Theorem \ref{c567gen} are instances of Theorem \ref{c67gen}, which is the main result of Section 4.

\begin{theorem}
\label{c67gen} The following problem can be solved in polynomial time:\\
Input: A graph $G \in \mathcal{G}(\widehat{C}_{6},\widehat{C}_{7})$, and
an induced complete bipartite subgraph $B$.\\
Question: Is $B$ generating?
\end{theorem}

A relating edge is a restricted case of a genereting subgraph. However, the complexity of the algorithm for recognizing related edges, presented in Section 3, is $O\left(  \left\vert V\right\vert \left(  \left\vert V\right\vert +\left\vert E\right\vert \right)  \right)  $, while the complexity of the algorithm which recognizes generating subgraphs in Section 4 is $O\left(  \left\vert V\right\vert ^{2} \left(  \left\vert V\right\vert +\left\vert E\right\vert \right)  \right)  $.

\section{Extendable Vertices}
An {\it extendable} vertex is a vertex $v \in V(G)$ such that there does not exist an independent set in $N_{2}(v)$ which dominates $N(v)$. This notion was first introduced in \cite{fhn:wcg5}. If $v$ is not extendable, a {\it witness} for non-extendability is a an independent set $S \subseteq N_{2}(v)$ such that $N(v) \subseteq N[S]$.

The main result of this section is a polynomial time algorithm which solves the following problem:\\
{\it Input:} A graph $G \in \mathcal{G}(\widehat{C}_{6}, \widehat{C}_{7})$ and a vertex $x \in V(G)$.\\
{\it Question:} Is $x$ extendable? \\

\subsection{Graphs Without Cycles of length 6}
In this subsection $G$ is a graph without cycles of length 6, and $x$ is a vertex in the graph, i.e. $G \in \mathcal{G}(\widehat{C}_{6})$ and $x \in V(G)$.

\begin{lemma}
\label{p3endpoints} Let $(a, b, c)$ be a path in $G[N_{2}(x)]$. Then there exists a vertex, $v$, such that $\{v\} = N(x) \cap N(a) = N(x) \cap N(c)$.
\end{lemma}

\begin{proof}
Assume on the contrary that Lemma \ref{p3endpoints} does not hold. Then there exist two distinct vertices, $v_{1}$ and $v_{2}$, such that $v_{1} \in N(x) \cap N(a)$ and $v_{2} \in N(x) \cap N(c)$. Therefore, $(x, v_{1}, a, b, c, v_{2})$ is a cycle of length 6, which is a contradiction.
\end{proof}

\begin{lemma}
\label{evenpath} Let $P=(a_{1},...,a_{2k+1})$, $k \geq 1$, be a path (not necessarily simple) in $G[N_{2}(x)]$. Then there exists a vertex, $v$, such that $\{v\} = N(x) \cap N(a_{1}) = N(x) \cap N(a_{2k+1})$.
\end{lemma}

\begin{proof}
By induction on $k$. If $k = 1$ then Lemma \ref{evenpath} is equivalent to Lemma \ref{p3endpoints}.

Assume by induction that Lemma \ref{evenpath} holds for $k$. We prove that it holds also for $k+1$. Let $P=(a_{1},...,a_{2k+3})$ be a path in $G[N_{2}(x)]$. By the induction hypothesis, there exists a vertex, $v$, such that $\{v\} = N(x) \cap N(a_{1}) = N(x) \cap N(a_{2k+1})$. Considering the subpath $P' = (a_{2k+1},a_{2k+2},a_{2k+3})$, Lemma \ref{p3endpoints} implies that $N(x) \cap N(a_{2k+1}) = N(x) \cap N(a_{2k+3})$. Therefore, $\{v\} = N(x) \cap N(a_{1}) = N(x) \cap N(a_{2k+3})$.
\end{proof}

\begin{lemma}
\label{codd} Let $A$ be a connected component of $G[N_{2}(x)]$ which contains an odd cycle. Then $|N(x) \cap N(V(A))|=1$.
\end{lemma}

\begin{proof}
Let $a$ be a vertex belonging to an odd cycle on $A$. For every vertex $b$ in $A$, there exists a path $P_{b}$ in $A$ with even number of edges connecting $a$ and $b$. By Lemma \ref{evenpath}, there exists a vertex, $v$, such that $\{v\} = N(x) \cap N(a) = N(x) \cap N(b)$. Hence, $\{v\} = N(x) \cap N(V(A))$.
\end{proof}

\begin{lemma}
\label{cbipartite6} Let $A$ be a bipartite connected component of $G[N_{2}(x)]$ with vertex sets of bipartition $V_{1}$ and $V_{2}$. Then for each $1 \leq i \leq 2$, if $|V_{i}| \geq 2$ then $|N(x) \cap N(V_{i})|=1$.
\end{lemma}

\begin{proof}
Let $1 \leq i \leq 2$ and $a \in V_{i}$. For every vertex $a \neq a' \in V_{i}$, there exists a path in $A$ with even number of edges connecting $a$ and $a'$. By Lemma \ref{evenpath}, there exists a vertex, $v$, such that $\{v\} = N(x) \cap N(a) = N(x) \cap N(a')$. Therefore, $\{v\} = N(x) \cap N(V_{i})$.
\end{proof}

\begin{lemma}
\label{cbigbipartite6} Let $A$ be a bipartite connected component of $G[N_{2}(x)]$ with vertex sets of bipartition $V_{1}$ and $V_{2}$, such that $min(|V_{1}|, |V_{2}|) \geq 2$. Then $|N(x) \cap N(V(A))|=1$.
\end{lemma}

\begin{proof}
$A$ contains a path $(a_{1}, a_{2}, a_{3}, a_{4})$, where $a_{1} \in V_{1}$ and $a_{4} \in V_{2}$. By Lemma \ref{cbipartite6}, for each $1 \leq i \leq 2$ there exists a vertex $v_{i}$, such that $\{v_{i}\} = N(x) \cap N(V_{i})$. Assume on the contrary that $v_{1} \neq v_{2}$. The cycle $(a_{1}, a_{2}, v_{2}, a_{4}, a_{3}, v_{1})$ is of length 6, which is a contradiction. Therefore, $v_{1} = v_{2}$, and $|N(x) \cap N(V(A))| = 1$.
\end{proof}

\begin{corollary}
\label{n2types} Let $A$ be a connected component of $G[N_{2}(x)]$. Then at least one of the following options holds. (See Fig. \ref{fign2types}.)
\begin{enumerate}

\item $|V(A)| = 1$.

\item $|N(x) \cap N(V(A))| = 1$.

\item $A$ is $K_{1,r}$, for $r \geq 1$. 

\end{enumerate}
\end{corollary}

\begin{proof}
Assume that the first 2 options of Corollary \ref{n2types} do not hold for $A$. By Lemma \ref{codd}, $A$ is bipartite. By Lemma \ref{cbigbipartite6}, at least one of the vertex sets of bipartition of $A$ contains only one vertex. The third option of Corollary \ref{n2types} holds.
\end{proof}

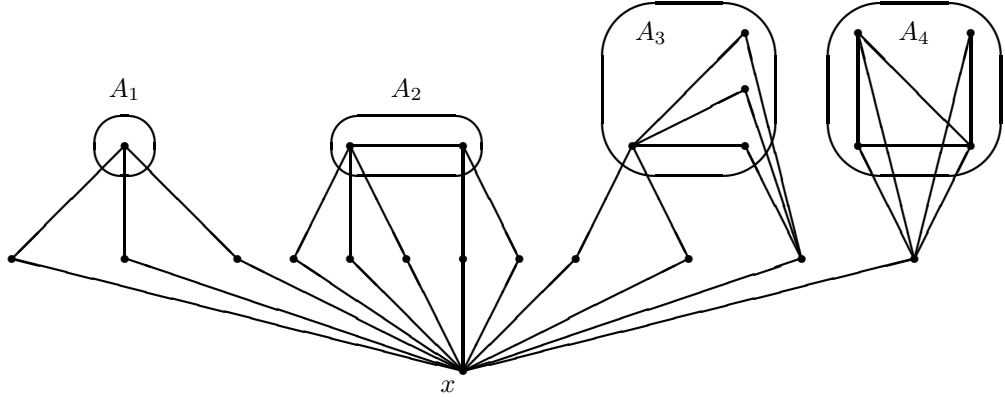
\begin{figure}[ht]
\setlength{\unitlength}{1.0cm} \begin{picture}(15,6)\thicklines
\put(7,0.5){\circle*{0.1}}
\put(2.5,3.5){\circle*{0.1}}
\put(12.25,3.5){\circle*{0.1}}
\put(13.75,3.5){\circle*{0.1}}
\put(12.25,5){\circle*{0.1}}
\put(13.75,5){\circle*{0.1}}
\put(9.25,3.5){\circle*{0.1}}
\put(10.75,3.5){\circle*{0.1}}
\put(10.75,4.25){\circle*{0.1}}
\put(10.75,5){\circle*{0.1}}
\multiput(1,2)(1.5,0){9}{\circle*{0.1}}
\multiput(5.5,3.5)(1.5,0){2}{\circle*{0.1}}
\put(6.8,0.3){\makebox(0,0){$x$}}
\put(7,0.5){\line(0,1){1.5}}
\put(7,0.5){\line(-4,1){6}}
\put(7,0.5){\line(-3,1){4.5}}
\put(7,0.5){\line(-2,1){3}}
\put(7,0.5){\line(-1,1){1.5}}
\put(7,0.5){\line(1,1){1.5}}
\put(7,0.5){\line(2,1){3}}
\put(7,0.5){\line(3,1){4.5}}
\put(7,0.5){\line(4,1){6}}
\put(7,0.5){\line(-1,2){0.75}}
\put(7,0.5){\line(1,2){0.75}}
\put(7,0.5){\line(-3,2){2.25}}
\put(2.5,3.5){\line(-1,-1){1.5}}
\put(2.5,3.5){\line(0,-1){1.5}}
\put(2.5,3.5){\line(1,-1){1.5}}
\put(12.25,3.5){\line(0,1){1.5}}
\put(12.25,3.5){\line(1,0){1.5}}
\put(13.75,3.5){\line(0,1){1.5}}
\put(13.75,3.5){\line(-1,1){1.5}}
\put(13,2){\line(-1,2){0.75}}
\put(13,2){\line(1,2){0.75}}
\put(13,2){\line(-1,4){0.75}}
\put(13,2){\line(1,4){0.75}}
\put(9.25,3.5){\line(1,0){1.5}}
\put(9.25,3.5){\line(2,1){1.5}}
\put(9.25,3.5){\line(1,1){1.5}}
\put(9.25,3.5){\line(-1,-2){0.75}}
\put(9.25,3.5){\line(1,-2){0.75}}
\put(11.5,2){\line(-1,2){0.75}}
\put(11.5,2){\line(-1,4){0.75}}
\put(11.5,2){\line(-1,3){0.75}}
\put(5.5,3.5){\line(1,0){1.5}}
\put(5.5,3.5){\line(0,-1){1.5}}
\put(5.5,3.5){\line(-1,-2){0.75}}
\put(5.5,3.5){\line(1,-2){0.75}}
\put(7,3.5){\line(0,-1){1.5}}
\put(7,3.5){\line(1,-2){0.75}}
\multiput(4.75,2)(1.5,0){3}{\circle*{0.1}}
\put(2.5,3.5){\oval(0.8,0.8)}
\put(6.25,3.5){\oval(2,0.8)}
\put(10,4.25){\oval(2.3,2.3)}
\put(13,4.25){\oval(2.3,2.3)}
\put(2.5,4.25){\makebox(0,0){$A_{1}$}}
\put(6.25,4.25){\makebox(0,0){$A_{2}$}}
\put(9.5,5){\makebox(0,0){$A_{3}$}}
\put(13,5){\makebox(0,0){$A_{4}$}}
\end{picture}
\caption{Distinct types of connected components of $N_{2}(x)$. $A_{1}$ contains only one vertex. $A_{4}$ contains an odd cycle, and therefore dominates only one vertex of $N(x)$. $A_{3}$ is bipartite with vertex sets of bipartition $V_{1}$ and $V_{2}$, $|V_{1}| = 1$, $|V_{2}| > 1$, and $V_{2}$ dominates only one vertex of $N(x)$. $A_{2}$ is $K_{1,1}$. }
\label{fign2types}
\end{figure}

Let $A^{*}$ be the set of all connected components $A$ of $G[N_{2}(x)]$ such that there exists a vertex $a \in V(A)$ for which $N(x) \cap N(a) = N(x) \cap N(V(A))$. For example, in Fig. \ref{fign2types}, $A_{1}$ and $A_{4}$ belong to $A^{*}$, while $A_{2}$ and $A_{3}$ do not. Define $H_{x} = G[N_{2}[x] \setminus N[V(A^{*})] ]$. Since $G \in \mathcal{G}(\widehat{C}_{6})$ and $H_{x}$ is an induced subgraph of $G$, also $H_{x} \in \mathcal{G}(\widehat{C}_{6})$.

\begin{corollary}
\label{gxystar} In the graph $H_{x}$ every connected component of $G[N_{2}(x)]$ is $K_{1,r}$ for $r \geq 1$. 
\end{corollary}

\begin{proof}
Follows immediately from the construction of $H_{x}$ and Corollary \ref{n2types}.
\end{proof}

\begin{lemma}
\label{omit} The following conditions are equivalent.
\begin{enumerate}
\item x is not extendable in G.
\item x is not extendable in $H_{x}$.
\end{enumerate}
\end{lemma}

\begin{proof}

$1 \implies 2$: Let $S \subseteq V(G)$ be an independent set in $N_{2}(x)$ which dominates $N(x)$. Let $S'$ be the set of all vertices of $S$ which exist also in the graph $H_{x}$, i.e. $S' = S \cap V(H_{x})$. The set $S'$ is a witness that $x$ is not extendable in $H_{x}$.

$2 \implies 1$: Let $S' \subseteq V(H_{x})$ be a witness that $x$ is not extendable in $H_{x}$.
For every component $A$ of $A^{*}$, choose a vertex $v_{A} \in V(A)$ such that $N(x) \cap N(v_{A}) = N(x) \cap N(V(A))$. Define $S = S' \cup \{v_{A} | A \in A^{*} \}$. Then $S$ is a witness that $x$ is not extendable in $G$.
\end{proof}

\begin{lemma}
\label{gxy2dom} Let $A$ be a connected component of $G[N_{2}(x)]$ in $H_{x}$, and $a$, $b$ two adjacent vertices in $A$. Then the following conditions hold.
\begin{enumerate}
\item $N(x) \cap N(V(A)) = N(x) \cap N(\{a,b\})$
\item $( N(a) \setminus N(b) ) \cap N(x) \neq \emptyset$
\item $( N(b) \setminus N(a) ) \cap N(x) \neq \emptyset$
\item $N(a) \cap N(b) \cap N(x) = \emptyset$
\end{enumerate}
\end{lemma}

\begin{proof}
If there exists a vertex $z \in V(A) \setminus \{a,b\}$, then there exists a path with two edges connecting $z$ to one of $a$ or $b$. By Lemma \ref{p3endpoints}, $N(x) \cap N(z) = N(x) \cap N(a)$ or $N(x) \cap N(z) = N(x) \cap N(b)$. Therefore, $N(x) \cap N(V(A)) = N(x) \cap N(\{a,b\})$.

If one of Conditions 2 and 3 does not hold, then there exists a vertex $v_{A} \in \{a,b\} \subseteq V(A)$ such that $N(x) \cap N(V(A)) = N(x) \cap N(v_{A})$, which contradicts the construction of $H_{x}$. Therefore, Conditions 2 and 3 hold.

It follows from Conditions 2 and 3 that there exist vertices $v_{a} \in N(x) \cap ( N(a) \setminus N(b) )$ and $v_{b} \in N(x) \cap ( N(b) \setminus N(a) )$. If there existed a vertex $v \in N(a) \cap N(b) \cap N(x)$ then $(x, v_{a}, a, v, b, v_{b})$ was a cycle of length 6, which was a contradiction. Therefore, Condition 4 holds.
\end{proof}

\subsection{Graphs Without Cycles of Lengths 6 and 7}
In this subsection $G \in \mathcal{G}(\widehat{C}_{6}, \widehat{C}_{7})$, $x \in V(G)$,  $H_{x} = G[N_{2}[x] \setminus N[V(A^{*})] ]$, and $A$ is a connected component of $N_{2}(x)$ in $H_{x}$. The vertices $a$ and $b$ are adjacent to each other, and belong to $A$. 

\begin{lemma}
\label{star} If $N(a) \cap N(x) = \{v_{1},...,v_{k}\}$ when $k > 2$, then $N(\{v_{1},...,v_{k}\}) \cap N_{2}(x) = \{a\}$.
\end{lemma}

\begin{proof}
Corollary \ref{gxystar} implies that $A$ is $K_{1,r}$, for $r \geq 1$. By Lemma \ref{cbipartite6}, $a$ is adjacent to all vertices of $V(A) \setminus \{a\}$. Assume on the contrary that Lemma \ref{star} does not hold. There exists a vertex $a' \in ( N(v_{1}) \cap N_{2}(x) ) \setminus \{a\}$. Lemma \ref{gxy2dom} implies that $a' \not\in V(A)$. Let $A'$ be the connected component of $N_{2}(x)$ in $H_{x}$ which contains $a'$. There exists a vertex $a'' \in V(A') \cap N(a')$.

Lemma \ref{gxy2dom} implies that there exists a vertex $v \in ( N(a'') \setminus N(a') ) \cap N(x)$. If $v \neq v_{2}$ then $(x,v,a'',a',v_{1},a,v_{2})$ is a cycle of length 7. Otherwise, $(x,v_{2},a'',a',v_{1},a,v_{3})$ is a cycle of length 7. (See Fig. \ref{figk1312tmp}.) In both cases we obtained a contradiction. Therefore, Lemma \ref{star} holds.
\end{proof}

\begin{figure}[h]
\setlength{\unitlength}{0.9cm} \begin{picture}(15,4)\thicklines
\put(3.5,0.5){\circle*{0.1}}
\multiput(0.5,2)(1.5,0){5}{\circle*{0.1}}
\multiput(1.25,3.5)(1.5,0){4}{\circle*{0.1}}
\put(1.25,3.8){\makebox(0,0){$b$}}
\put(2.75,3.8){\makebox(0,0){$a$}}
\put(4.25,3.8){\makebox(0,0){$a'$}}
\put(5.75,3.8){\makebox(0,0){$a''$}}
\put(2,3.9){\makebox(0,0){$A$}}
\put(5,3.9){\makebox(0,0){$A'$}}
\put(1.25,3){\circle*{0.1}}
\put(1.25,3){\line(3,1){1.5}}
\put(1.25,3){\line(-3,-4){0.75}}
\put(3.3,0.4){\makebox(0,0){$x$}}
\put(0.3,1.8){\makebox(0,0){$v_{b}$}}
\put(1.8,1.8){\makebox(0,0){$v_{3}$}}
\put(3.3,1.8){\makebox(0,0){$v_{2}$}}
\put(5.1,1.8){\makebox(0,0){$v_{1}$}}
\put(6.6,1.8){\makebox(0,0){$v$}}
\put(1.25,3.5){\line(1,0){1.5}}
\put(4.25,3.5){\line(1,0){1.5}}
\put(3.5,0.5){\line(-2,1){3}}
\put(3.5,0.5){\line(-1,1){1.5}}
\put(3.5,0.5){\line(0,1){1.5}}
\put(3.5,0.5){\line(1,1){1.5}}
\put(3.5,0.5){\line(2,1){3}}
\put(0.5,2){\line(1,2){0.75}}
\put(2,2){\line(1,2){0.75}}
\put(3.5,2){\line(-1,2){0.75}}
\put(5,2){\line(-3,2){2.25}}
\multiput(3.5,2)(0.6,0.4){4}{\line(3,2){0.4}} 
\put(5,2){\line(-1,2){0.75}}
\multiput(6.5,2)(-0.28,0.56){3}{\line(-1,2){0.2}} 
\put(11,0.5){\circle*{0.1}}
\multiput(9.5,2)(1.5,0){4}{\circle*{0.1}}
\multiput(8,3.5)(1.5,0){6}{\circle*{0.1}}
\multiput(8,3.5)(3,0){3}{\line(1,0){1.5}}
\put(8,3.8){\makebox(0,0){$b$}}
\put(9.5,3.8){\makebox(0,0){$a$}}
\put(11,3.8){\makebox(0,0){$a'$}}
\put(12.5,3.8){\makebox(0,0){$a''$}}
\put(14,3.8){\makebox(0,0){$z'$}}
\put(15.5,3.8){\makebox(0,0){$z''$}}
\put(8.75,3.9){\makebox(0,0){$A$}}
\put(11.75,3.9){\makebox(0,0){$A'$}}
\put(14.75,3.9){\makebox(0,0){$A''$}}
\put(10.8,0.4){\makebox(0,0){$x$}}
\put(9.3,1.8){\makebox(0,0){$v_{b}$}}
\put(10.8,1.8){\makebox(0,0){$v_{1}$}}
\put(12.6,1.8){\makebox(0,0){$v_{2}$}}
\put(14.1,1.8){\makebox(0,0){$v$}}
\put(11,0.5){\line(-1,1){1.5}}
\put(11,0.5){\line(0,1){1.5}}
\put(11,0.5){\line(1,1){1.5}}
\put(11,0.5){\line(2,1){3}}
\put(9.5,2){\line(-1,1){1.5}}
\put(11,2){\line(-1,1){1.5}}
\put(12.5,2){\line(-2,1){3}}
\put(11,2){\line(0,1){1.5}}
\put(12.5,2){\line(0,1){1.5}}
\put(11,2){\line(2,1){3}}
\put(12.5,2){\line(2,1){3}}
\multiput(12.5,2)(-0.5,0.5){3}{\line(-1,1){0.4}} 
\multiput(14,2)(-0.5,0.5){3}{\line(-1,1){0.4}} 
\multiput(14,2)(-0.52,0.26){6}{\line(-2,1){0.4}} 
\end{picture}
\caption{Proofs of Lemmas \ref{star} (left) and \ref{star2} (right). The dashed edges are not in the graph.}
\label{figk1312tmp}
\end{figure}
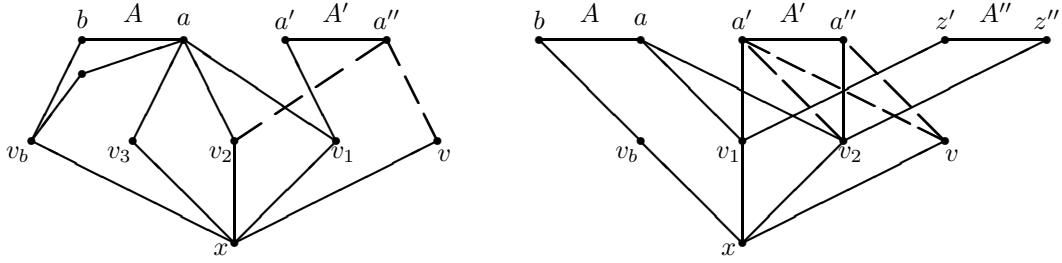

\begin{lemma}
\label{star2} If $N(a) \cap N(x) = \{v_{1},v_{2}\}$ then one of the following two options holds.
\begin{itemize}
\item $N(\{v_{1},v_{2}\}) \cap N_{2}(x) = \{a\}$.
\item There exists exactly one connected component $A' \neq A$ of $N_{2}(x)$ such that $\{v_{1},v_{2}\} = N(V(A')) \cap N(x)$. For every connected component $A \neq A'' \neq A'$ of $N_{2}(x)$, it holds that $\{v_{1},v_{2}\} \cap N(V(A'')) = \emptyset$.
\end{itemize}
\end{lemma}

\begin{proof}
Clearly, $A$ is $K_{1,r}$, for some $r \geq 1$, and $a$ is adjacent to all vertices of $V(A) \setminus \{a\}$. By Lemma \ref{gxy2dom}, there exists a vertex $v_{b} \in (N(b) \setminus N(a)) \cap N(x)$. 

Assume that the first option of Lemma \ref{star2} does not hold. There exists a vertex $a' \in ( N_{2}(x) \cap N(v_{1}) ) \setminus \{ a \}$. Lemma \ref{gxy2dom} implies that $a' \not\in V(A)$. Let $A'$ be the connected component of $N_{2}(x)$ which contains $a'$, and let $a'' \in N_{2}(x) \cap N(a')$.

We prove that $N(a'') \cap N(x) = \{v_{2}\}$. Lemma \ref{gxy2dom} implies that there exists a vertex $v \in ( N(a'') \cap N(x) ) \setminus \{v_{1}\}$. Assume on the contrary that $v \neq v_{2}$. Then $(x, v, a'', a', v_{1}, a, v_{2})$ is a cycle of length 7, which is a contradiction. (See Fig. \ref{figk1312tmp}.) Therefore, $N(a'') \cap N(x) = \{v_{2}\}$.

We prove that $N(a') \cap N(x) = \{v_{1}\}$. By Lemma \ref{gxy2dom}, $v_{2} \not\in N(a')$. If there exists a vertex $v \in ( N(x) \cap N(a') ) \setminus \{v_{1},v_{2}\}$ then $(x, v, a', v_{1}, a, v_{2})$ is a cycle of length 6, which is a contradiction. Therefore, $N(a') \cap N(x) = \{v_{1}\}$. 

Lemma \ref{gxy2dom} implies that $N(V(A')) \cap N(x) = N( \{ a', a'' \} ) \cap N(x) = \{v_{1},v_{2}\}$. Assume on the contrary that there exits a connected component $A \neq A'' \neq A'$ of $N_{2}(x)$ with vertices $z'$ and $z''$, which are adjacent to $v_{1}$ and $v_{2}$, respectively. Then $(z', v_{1}, a', a'', v_{2}, z'')$ is a cycle of length 6, which is a contradiction. The second option of Lemma \ref{star2} holds.
\end{proof}

\begin{corollary}
\label{star2dominate} If $|N(a) \cap N(x)| = 2$ then in the graph $H_{x}$ every independent set of $N_{2}(x)$ which dominates $N(x)$, contains the vertex $a$.
\end{corollary}

\begin{proof} 
Denote $N(a) \cap N(x) = \{v_{1},v_{2}\}$. Let $S$ be a maximal independent set of $N_{2}(x)$ which dominates $N(x)$. If $N(\{v_{1},v_{2}\}) \cap N_{2}(x) = \{a\}$ then obviously $a \in S$.

If $N(\{v_{1},v_{2}\}) \cap N_{2}(x) \neq \{a\}$ then, by Lemma \ref{star2}, there exits only one component $A' \neq A$ of $N_{2}(x)$ with vertices adjacent to $\{v_{1},v_{2}\}$. However, no independent set in $A'$ dominates both $v_{1}$ and $v_{2}$. Therefore, $a \in S$.
\end{proof}

\begin{corollary}
\label{2stars} 
If $|N(a) \cap N(x)| = |N(b) \cap N(x)| = 2$ then in the graph $H_{x}$ there does not exist an independent set of $N_{2}(x)$ which dominates $N(x)$.
\end{corollary}

\begin{proof}
Suppose on the contrary that $H_{x}$ contains an independent set $S \subseteq N_{2}(x)$ which dominates $N(x)$. By Corollary \ref{star2dominate}, $S$ contains both $a$ and $b$, which contradicts the fact that $S$ is independent.  
\end{proof}

Let $S^{*} = \{ v \in V(H_{x}) \ | \ v \in N_{2}(x), \ |N(v) \cap N(x)| > 1 \}$. By Corollary \ref{star2dominate}, every independent set in $N_{2}(x)$ which dominates $N(x)$ contains $S^{*}$. Therefore, if $S^{*}$ is not independent then there does not exist an independent set in $N_{2}(x)$ which dominates $N(x)$. In this case $x$ is extendable in $H_{x}$ and, by Lemma \ref{omit}, $x$ is also extendable in $G$. 

Define $H^{*}_{x}$ to be the induced subgraph of $H_{x}$ with vertex set $V(H^{*}_{x})=V(H_{x}) \setminus N[S^{*}]$.

\begin{lemma}
\label{Gstarequiv} Suppose $S^{*}$ is independent. Then x is extendable in $H_{x}$ if and only if it is extendable in $H^{*}_{x}$. 
\end{lemma}

\begin{proof} Let $S_{1} \subseteq V(H_{x})$ be an independent set in $N_{2}(x)$ which dominates $N(x)$. Corollary \ref{star2dominate} implies that $S^{*} \subseteq S_{1}$. Define $S_{2} \subseteq V(H^{*}_{x})$ by $S_{2} = S_{1} \setminus S^{*}$. Then in the graph $H^{*}_{x}$ the set $S_{2}$ is independent, contained in $N_{2}(x)$, and dominates $N(x)$. 

Let $S_{2} \subseteq V(H^{*}_{x})$ be an independent set in $N_{2}(x)$ which dominates $N(x)$. Define $S_{1} \subseteq V(H_{x})$ by $S_{1} = S_{2} \cup S^{*}$. Then $S_{1}$ is an independent set in $N_{2}(x)$ which dominates $N(x)$. 
\end{proof}

\begin{theorem}
\label{c6c7pol} The following problem can be solved in $O\left(  \left\vert V\right\vert \left(  \left\vert V\right\vert +\left\vert E\right\vert \right)  \right)  $ time.\\
Input: A graph $G \in \mathcal{G}(\widehat{C}_{6},\widehat{C}_{7})$ and a vertex $x \in V(G)$. \\
Question: Is x extendable?
\end{theorem}

\begin{proof}
Let $A$ be a connected component of $N_{2}(x)$ in the graph $H^{*}_{x}$. Then $A$ is $K_{1,r}$ for some $r \geq 1$. Moreover, $A$ is adjacent to exactly 2 vertices of $N(x)$, because otherwise the vertices of $A$ were in $N[S^{*}]$, and not in $V(H^{*}_{x})$. Every vertex set of bipartition of $A$ dominates exactly one vertex of $N(x)$. In order to decide whether $x$ is extendable in $H^{*}_{x}$, we define a flow network. We use the same technique as in Theorem 2.2 of \cite{lt:relatedc4}.

Let $A_{1},...,A_{k}$ be the connected components of $N_{2}(x)$ in the graph $H^{*}_{x}$. Define a flow network $F_{x} = \{G_{F} = (V_{F},E_{F}) , s \in V_{F}, t \in V_{F}, c : E_{F} \longrightarrow \mathbb{R} \}$ as follows. (See Fig. \ref{figflow}.) Let $V_{F} = N_{1}(x) \cup N_{2}(x) \cup \{ z_{1},..., z_{k}, s, t\}$, where $z_{1}, ..., z_{k}, s, t$ are new vertices, $s$ and $t$ are the source and sink of the network, respectively. The directed edges $E_{F}$ are:

\begin{itemize}

\item the directed edges $sz_{i}$, for each $1 \leq i \leq k$;

\item the directed edges $z_{i}a$, for each $1 \leq i \leq k$ and for each $a \in A_{i}$;

\item all directed edges $v_{2}v_{1}$ such that $v_{2} \in N_{2}(x)$, $v_{1} \in N(x)$ and $v_{1}v_{2}\in E(H^{*}_{x})$;

\item the directed edges $vs$, for each $v \in N(x)$;

\end{itemize}

Let $c\equiv1$. Invoke any polynomial time algorithm for finding a maximum flow $f : E_{F}\longrightarrow \mathbb{R}$ in the network, for example Ford and Fulkerson's algorithm. The flow in a vertex $v \in V_{F}$ is defined by: $\Sigma_{(u,v)\in E_{F}}f(u,v)$. 

\begin{figure}[h]
\setlength{\unitlength}{1.0cm} \begin{picture}(20,9)\thicklines
\put(4.5,5){\circle*{0.1}}
\multiput(3,2)(0,1.5){5}{\circle*{0.1}}
\multiput(1.5,0.5)(0,2.25){4}{\circle*{0.1}}
\multiput(1.5,2)(0,2.25){4}{\circle*{0.1}}
\put(4.7,5){\makebox(0,0){$x$}}
\put(4.5,5){\line(-1,0){1.5}}
\put(4.5,5){\line(-1,-1){1.5}}
\put(4.5,5){\line(-1,-2){1.5}}
\put(4.5,5){\line(-1,1){1.5}}
\put(4.5,5){\line(-1,2){1.5}}
\multiput(1.5,0.5)(0,2.25){4}{\line(0,1){1.5}}
\put(1.5,0.5){\line(1,1){1.5}}
\put(1.5,2){\line(1,1){1.5}}
\put(1,0.5){\line(4,3){2}}
\put(0.5,0.5){\line(5,3){2.5}}
\put(1,0.5){\line(1,3){0.5}}
\put(0.5,0.5){\line(2,3){1}}
\put(1,0.5){\circle*{0.1}}
\put(0.5,0.5){\circle*{0.1}}
\put(1.5,2.75){\line(2,-1){1.5}}
\put(1.5,4.25){\line(2,1){1.5}}
\put(1.5,5){\line(1,0){1.5}}
\put(1.5,6.5){\line(1,0){1.5}}
\put(1.5,7.25){\line(2,-3){1.5}}
\put(1.5,8.75){\line(2,-1){1.5}}
\put(0.75,8.75){\line(3,-1){2.25}}
\put(0.75,8.75){\line(1,-2){0.75}}
\put(0.75,8.75){\circle*{0.1}}
\put(13.5,5){\circle*{0.1}}
\multiput(12,2)(0,1.5){5}{\circle*{0.1}}
\multiput(10.5,0.5)(0,2.25){4}{\circle*{0.1}}
\multiput(10.5,2)(0,2.25){4}{\circle*{0.1}}
\put(13.7,5){\makebox(0,0){$t$}}
\put(12,5){\vector(1,0){1.5}}
\put(12,3.5){\vector(1,1){1.5}}
\put(12,2){\vector(1,2){1.5}}
\put(12,6.5){\vector(1,-1){1.5}}
\put(12,8){\vector(1,-2){1.5}}
\put(10.5,0.5){\vector(1,1){1.5}}
\put(10.5,2){\vector(1,1){1.5}}
\put(10,0.5){\vector(4,3){2}}
\put(9.5,0.5){\line(5,3){2.5}}
\put(10,0.5){\circle*{0.1}}
\put(9.5,0.5){\circle*{0.1}}
\put(10.5,2.75){\vector(2,-1){1.5}}
\put(10.5,4.25){\vector(2,1){1.5}}
\put(10.5,5){\vector(1,0){1.5}}
\put(10.5,6.5){\vector(1,0){1.5}}
\put(10.5,7.25){\vector(2,-3){1.5}}
\put(10.5,8.75){\vector(2,-1){1.5}}
\put(9.75,8.75){\vector(3,-1){2.25}}
%\put(9.75,8.75){\vector(1,-2){0.75}}
\put(9.75,8.75){\circle*{0.1}}
\multiput(9,1.25)(0,2.25){4}{\circle*{0.1}}
\multiput(9,1.25)(0,2.25){4}{\vector(2,1){1.5}}
\multiput(9,1.25)(0,2.25){4}{\vector(2,-1){1.5}}
\put(9,1.25){\vector(2,-3){0.5}}
\put(9,1.25){\vector(4,-3){1}}
\put(9,8){\vector(1,1){0.75}}
\put(7.875,4.625){\circle*{0.1}}
\put(7.875,4.625){\vector(1,1){1.125}}
\put(7.875,4.625){\vector(1,-1){1.125}}
\put(7.875,4.625){\vector(1,3){1.125}}
\put(7.875,4.625){\vector(1,-3){1.125}}
\put(7.675,4.625){\makebox(0,0){$s$}}
\put(0.7,1.25){\makebox(0,0){$A_{1}$}}
\put(1.2,3.5){\makebox(0,0){$A_{2}$}}
\put(1.2,5.75){\makebox(0,0){$A_{3}$}}
\put(0.9,8){\makebox(0,0){$A_{4}$}}
\put(8.7,1.25){\makebox(0,0){$z_{1}$}}
\put(9,3.8){\makebox(0,0){$z_{2}$}}
\put(9,5.4){\makebox(0,0){$z_{3}$}}
\put(8.7,8){\makebox(0,0){$z_{4}$}}
\put(5.5,5){\line(1,0){1.5}}
\put(6.5,5.5){\line(1,-1){0.5}}
\put(6.5,4.5){\line(1,1){0.5}}
\end{picture}
\caption{An example of the construction of the flow network $F_{x}$ (right) from $H^{*}_{x}$ (left).}
\label{figflow}
\end{figure}
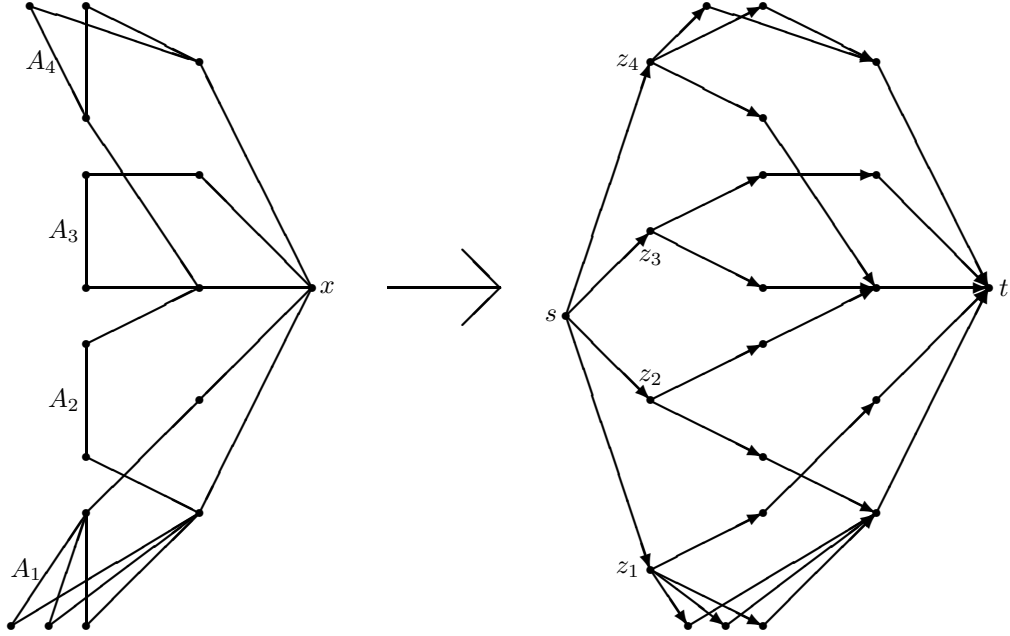

Let $S$ be the set of vertices in $N_{2}(x)$ in which there is a positive flow. It is easy to prove that $S$ is independent, and every augmenting path in Ford and Fulkerson's algorithm increases by one each of $|f|$, $|S|$ and $|N(x) \cap N(S)|$. Therefore, $|f| = |S| = |N(x) \cap N(S)|$. Moreover, for every independent set $S^{\prime}$ of $N_{2}(x)$, it holds that $|N(x) \cap N(S)| \geq |N(x) \cap N(S^{\prime})|$. For more details see \cite{lt:relatedc4}. 

If $S$ dominates $N(x)$ in $H^{*}_{x}$ then obviously $x$ is not extendable. Otherwise, there does not exist an independent set in $N_{2}(x)$ which dominates $N(x)$, and therefore $x$ is extendable.

The following polynomial algorithm receives as its input a graph $G \in \mathcal{G}(\widehat{C}_{6},\widehat{C}_{7})$ and a vertex $x \in V(G)$. If $x$ is extendable, the algorithm returns $\emptyset$. Otherwise, the algorithm returns a witness that $x$ is not extendable.

\begin{algorithm}[H]
\DontPrintSemicolon
{ $T \longleftarrow \emptyset$ \; } 
{ Find the connected components of $N_{2}(x)$ \; } 
{ Find $A^{*}$ \; } 
{ \For{{\bf each} $A \in A^{*}$ } 
{ { Choose $v_{A} \in V(A)$ such that $N(v_{A}) \cap N(x) = N(V(A)) \cap N(x)$ \; } 
{ $T \longleftarrow T \cup \{ v_{A} \} $ \; } 
} 
}
{ Construct the graph $H_{x}$.  \; }
{ $S^{*} \longleftarrow \{ a \in V(H_{x}) \ | \ a \in N_{2}(x), | N(a) \cap N(x)| > 1 \}$.  \; }
{ \If{$S^{*}$ is not independent }
{ \Return $\emptyset$.  \; }
}
{ Construct the graph $H^{*}_{x}$.  \; }
{ Construct the flow network $F_{x}$  \; }
{ Find a maximum flow $f_{x} : E_{F}\longrightarrow \mathbb{R}$ in the network $F_{x}$.  \; }
{ Let $S$ be the set of vertices in $N_{2}(x)$ in which there is a positive flow. \; }
{ \If{$S$ does not dominate $N(x)$ in $H^{*}_{x}$}
{ \Return $\emptyset$.  \;}
}
{ \Return $T \cup S \cup S^{*}$.  \; }
\caption{Decide whether $x$ is extendable in $G \in \mathcal{G}(\widehat{C}_{6},\widehat{C}_{7})$ . \label{alghalfrelating}}
\end{algorithm}\medskip

Correctness of Algorithm \ref{alghalfrelating}: Follows from previous lemmas.

Complexity analysis of Algorithm \ref{alghalfrelating}: Constructing the set $T$ includes finding $N_{2}(x)$, the connected components of $G[N_{2}(x)]$, and the set $A^{*}$. This can be implemented in $O(|V| + |E|)$ time. Also finding the induced subgraphs $H_{x}$ can be implemented in $O(|V| + |E|)$ time. Finding $S^{*}_{x}$ and deciding whether it is independent can be done in $O(|V| +|E|)$ time, as well. This is also the complexity for constructing the induced subgraph $H^{*}_{x}$, and the flow network $F_{x}$.

One iteration of Ford and Fulkerson's algorithm can be implemented in $O\left(  \left\vert V\right\vert
+\left\vert E\right\vert \right)  $ time. In each iteration the number of
vertices in $N_{2}(x)$ with a positive flow increases by $1$. Therefore, the
number of iterations can not exceed $\left\vert V\right\vert $, and Ford and
Fulkerson's algorithm terminates in $O\left(  \left\vert V\right\vert \left(
\left\vert V\right\vert +\left\vert E\right\vert \right)  \right)$ time. 

Deciding whether $S_{x}$ dominates $N(x)$ can be done in $O(|V| +|E|)$ time. The total complexity of Algorithm \ref{alghalfrelating} is $O\left(  \left\vert V\right\vert \left(  \left\vert V\right\vert +\left\vert
E\right\vert \right)  \right)  $.
\end{proof}

\section{Relating Edges}
In this section we present a polynomial algorithm for recognizing relating edges in graphs without cycles of lengths 6 and 7.

\begin{lemma}
\label{c6split} Let $G \in \mathcal{G}(\widehat{C}_{6}, \widehat{C}_{7})$ and $xy \in E(G)$. The following conditions are equivalent.
\begin{enumerate}

\item The edge $xy$ is relating.

\item There exist an independent set, $S_{x} \subseteq N_{2}(x) \setminus N(y)$, which dominates $N(x) \cap N_{2}(y)$, and an independent set, $S_{y} \subseteq N_{2}(y) \setminus N(x)$, which dominates $N(y) \cap N_{2}(x)$.
\end{enumerate}
\end{lemma}

\begin{proof}
$1 \implies 2$: Let $S$ be a witness that $xy$ is relating. Define $S_{x} = S \cap N_{2}(x)$ and $S_{y} = S \cap N_{2}(y)$. Clearly, $S_{x}$ dominates $N(x) \cap N_{2}(y)$ and $S_{y}$ dominates $N(y) \cap N_{2}(x)$.

$2 \implies 1$: Let $x'' \in S_{x}$ and $y'' \in S_{y}$. There exist vertices $x' \in N(x) \cap N(x'')$ and $y' \in N(y) \cap N(y'')$. If $x''$ and $y''$ were adjacent then $(x'',x',x,y,y',y'')$ was a cycle of length 6. Therefore, $S_{x} \cup S_{y}$ is independent. Let $S$ be a maximal independent set of $G[V(G) \setminus N[\{x,y\}]]$ which contains $S_{x} \cup S_{y}$. Then $S$ is a witness that $xy$ is relating.
\end{proof}

\begin{theorem}
\label{relatingpol} The following problem can be solved in $O\left(  \left\vert V\right\vert \left(  \left\vert V\right\vert +\left\vert E\right\vert \right)  \right)  $ time.\\
Input: A graph $G \in \mathcal{G}(\widehat{C}_{6},\widehat{C}_{7})$ and an edge $xy \in E(G)$. \\
Question: Is xy relating?
\end{theorem}

\begin{proof}
The following polynomial algorithm receives as its input a graph $G \in \mathcal{G}(\widehat{C}_{6},\widehat{C}_{7})$ and an edge $xy \in E(G)$. If $xy$ is relating, the algorithm returns an independent set in $N_{2}(\{x,y\})$ which dominates $N(x) \Delta N(y)$. Otherwise, $\emptyset$ is returned.

\begin{algorithm}[H]
\DontPrintSemicolon
{ $A \longleftarrow \emptyset$  \;  }
{ $B \longleftarrow \emptyset$  \;  }
{ \If{$N(x) \cap N_{2}(y) \neq \emptyset$}
{
{ $A \longleftarrow Alg1(G[(N_{2}(x) \cap N_{3}(y)) \cup ( N(x) \cap N_{2}(y) ) \cup \{x\} ],x)$  \; }
{ \If{$A=\emptyset$}
{ \Return $\emptyset$  \; }
}
}
}
{ \If{$N(y) \cap N_{2}(x) \neq \emptyset$}
{
{ $B \longleftarrow Alg1(G[(N_{2}(y) \cap N_{3}(x)) \cup ( N(y) \cap N_{2}(x) ) \cup \{y\} ],y)$  \; }
{ \If{$B=\emptyset$}
{ \Return $\emptyset$  \; }
}
}
}
{ \Return $A \cup B$.  \; }
\caption{Decide whether the edge $xy$ in the graph $G \in \mathcal{G}(\widehat{C}_{6},\widehat{C}_{7})$ is relating. \label{algrelating}}
\end{algorithm}\medskip

Correctness of Algorithm \ref{algrelating}: Follows from Lemma \ref{c6split}.

Complexity analysis of Algorithm \ref{algrelating}: The algorithm invokes Algorithm \ref{alghalfrelating} at most twice. Therefore, the complexity of Algorithm \ref{algrelating} is equal to the complexity of Algorithm \ref{alghalfrelating}, i.e. $O\left(  \left\vert V\right\vert \left(  \left\vert V\right\vert +\left\vert
E\right\vert \right)  \right)  $. 
\end{proof}

\section{Generating Subgraphs}
In this section we present a polynomial algorithm for recognizing generating subgraphs in $\mathcal{G}(\widehat{C}_{6},\widehat{C}_{7})$. Through this section $G \in \mathcal{G}(\widehat{C}_{6},\widehat{C}_{7})$, and $B$ is an induced complete bipartite subgraph of $G$.

\begin{lemma}
\label{disjoint} Let $a$ and $b$ be two distinct vertices in $V(B)$. Let $a' \in N(a) \setminus V(B)$ and $b' \in N(b) \setminus V(B)$ such that $a' \neq b'$. Let $a'' \in N_{2}(a) \cap N(a')$ and $b'' \in N_{2}(b) \cap N(b')$. Then $a''$ and $b''$ are not adjacent.
\end{lemma}

\begin{proof}
If $a$ and $b$ are adjacent, the lemma holds, or otherwise $(a'',a',a,b,b',b'')$ is a cycle of length 6. 

If $a$ and $b$ are not neighbors, there exists a vertex $c \in V(B) \cap N(a) \cap N(b)$. If $a'' \in N(b'')$ then $(a'',a',a,c,b,b',b'')$ is a cycle of length 7, which is a contradiction.
\end{proof}

For every vertex $b \in V(B)$, define $G_{b} = G[(N_{2}(b) \cap N_{2}(V(B))) \cup (N(b) \cap N(V(B))) \cup \{b\}]$. Since $G \in \mathcal{G}(\widehat{C}_{6},\widehat{C}_{7})$ and $G_{b}$ is an induced subgraph of $G$, also $G_{b} \in \mathcal{G}(\widehat{C}_{6},\widehat{C}_{7})$.

\begin{corollary}
\label{independent} Let $a$ and $b$ be two distinct vertices in $V(B)$. Let $S_{a} \subseteq V(G_{a}) \cap N_{2}(a)$ and $S_{b} \subseteq V(G_{b}) \cap N_{2}(b)$ be two independent sets. Then $S_{a} \cup S_{b}$ is an independent set in G.
\end{corollary}

\begin{lemma}
\label{generatinglemma} The following conditions are equivalent.
\begin{enumerate}

\item The subgraph $B$ is generating in G.

\item The vertex $b$ is not extendable in $G_{b}$, for every $b \in V(B)$.
\end{enumerate}
\end{lemma}

\begin{proof}

$1 \implies 2$: Let $S$ be a witness that $B$ is a generating subgraph of $G$. For every $b \in V(B)$ define $S_{b} = S \cap V(G_{b})$. Then $S_{b}$ is a witness that  $b$ is not extendable in $G_{b}$.

$2 \implies 1$: For every $b \in V(B)$, let $S_{b}$ be a witness that $b$ is not extendable in $G_{b}$. Define $S = \bigcup_{b \in V(B)} S_{b}$. By Corollary \ref{independent}, $S$ is independent. The set $S$ is a witness that $B$ is a generating subgraph of $G$.
\end{proof}

\begin{theorem}
\label{gengpol} The following problem can be solved in $O\left(  \left\vert V\right\vert ^{2} \left(  \left\vert V\right\vert +\left\vert E\right\vert \right)  \right)  $ time.\\
Input: A graph $G \in \mathcal{G}(\widehat{C}_{6},\widehat{C}_{7})$ and an induced bipartite subgraph $B$. \\
Question: Is B generating?
\end{theorem}

\begin{proof}
The following algorithm receives as its input a graph $G \in \mathcal{G}(\widehat{C}_{6},\widehat{C}_{7})$ and an induced complete bipartite subgraph, $B$. The algorithm returns a witness that $B$ is generating, if exists, and $\emptyset$ otherwise.

\begin{algorithm}[H]
\DontPrintSemicolon
{ \For{{\bf each} $b \in V(B)$ } 
{  
{ $S_{b} \longleftarrow Alg1(G_{b},b) $ \; } 
{ \If{$S_{b}=\emptyset$ and $|V(B)| > 1$}
{ \Return $\emptyset$  \; }
}
} 
}
{ $S \longleftarrow \bigcup_{b \in V(B)} S_{b}$.  \; }
{ Extract $S$ to a maximal independent set, $S^{*}$, of $G \setminus N[V(B)]$.  \; }
{ \Return $S^{*}$.  \; }
\caption{Decide whether the subgraph $B$ of $G \in \mathcal{G}(\widehat{C}_{6},\widehat{C}_{7})$ is generating. \label{alggenerating}}
\end{algorithm}\medskip

Correctness of Algorithm \ref{alggenerating}: Follows from Lemma \ref{generatinglemma}.

Complexity analysis of Algorithm \ref{alggenerating}: The complexity of Algorithm \ref{alghalfrelating} is $O\left(  \left\vert V\right\vert \left(  \left\vert V\right\vert +\left\vert E\right\vert \right)  \right)  $, and it is invoked at most $O(|V|)$ times. Therefore, the total complexity of Algorithm \ref{alggenerating} is $O\left(  \left\vert V\right\vert ^{2} \left(  \left\vert V\right\vert +\left\vert E\right\vert \right)  \right)  $.
\end{proof}

\section{Conclusions and Future Work}
We presented a polynomial algorithm for recognizing generating subgraphs in $\mathcal{G}(\widehat{C}_{6},\widehat{C}_{7})$. However, we did not find a polynomial algorithm which receives $G \in \mathcal{G}(\widehat{C}_{6},\widehat{C}_{7})$ as its input, and finds $WCW(G)$, the vector space of all weight functions $w$ such that $G$ is $w$-well-covered. Checking all induced complete bipartite subgraphs is obviously not a polynomial algorithm, since the number of checked subgraphs can be exponential. Nevertheless, we conjecture the following.

\begin{conjecture}
\label{wwc67pol} The following problem can be solved in polynomial time.\\
Input: A graph $G \in \mathcal{G}(\widehat{C}_{6},\widehat{C}_{7})$. \\
Output: $WCW(G)$.
\end{conjecture}

It is known that recognizing relating edges is a polynomial task for $\mathcal{G}(\widehat{C}_{4},\widehat{C}_{6})$ \cite{lt:relatedc4}, and for $\mathcal{G}(\widehat{C}_{5},\widehat{C}_{6})$ \cite{lt:wc456}. We proved that also for $\mathcal{G}(\widehat{C}_{6},\widehat{C}_{7})$ the problem is polynomial. However, in most parts of the proof, the only forbidden cycles were of length 6. Hence, we conjecture that all three theorems are instances of the following.

\begin{conjecture}
\label{polc6} The following problem is polynomial solvable:\\
Input: A graph $G \in \mathcal{G}(\widehat{C}_{6})$ and an edge $xy \in E(G)$. \\
Question: Is $xy$ relating?
\end{conjecture}


\begin{thebibliography}{99}                                                                                               %


\bibitem {bnz:related}J. I. Brown, R. J. Nowakowski, I. E. Zverovich,
\emph{The structure of well-covered graphs with no cycles of length }$4$,
Discrete Mathematics \textbf{307} (2007) 2235 -- 2245.

\bibitem {cer:degree}Y. Caro, N. Ellingham, G. F. Ramey, \emph{Local structure
when all maximal independent sets have equal weight}, SIAM Journal on Discrete
Mathematics \textbf{11} (1998) 644-654.

\bibitem {cst:structures}Y. Caro, A. Seb\H{o}, M. Tarsi, \emph{Recognizing
greedy structures}, Journal of Algorithms \textbf{20} (1996) 137-156.

\bibitem {cs:note}V. Chvatal, P. J. Slater, \emph{A note on well-covered
graphs}, Quo vadis, Graph Theory Ann Discr Math 55, North Holland, Amsterdam,
1993, 179-182.

\bibitem {fhn:wcg5}A. Finbow, B. Hartnell, R. Nowakowski, \emph{A
characterization of well-covered graphs of girth 5 or greater}, Journal of
Combinatorial Theory Ser. B. \textbf{57} (1993) 44-68.

\bibitem {fhn:wc45}A. Finbow, B. Hartnell, R. Nowakowski \emph{A
characterization of well-covered graphs that contain neither 4- nor 5-cycles},
Journal of Graph Theory \textbf{18} (1994) 713-721.

\bibitem {lt:wc4567}V. E. Levit, D. Tankus \emph{Weighted well-covered graphs
without $C_{4}$, $C_{5}$, $C_{6}$, $C_{7}$}, Discrete Applied Mathematics
\textbf{159} (2011) 354-359.

\bibitem {lt:relatedc4}V. E. Levit, D. Tankus, \emph{On relating edges in
graphs without cycles of length 4}, Journal of Discrete Algorithms \textbf{26}
(2014) 28-33.

\bibitem {lt:wc456}V. E. Levit, D. Tankus, \emph{Well-covered graphs without cycles of lengths 4, 5 and 6}, Discrete Applied Mathematics \textbf{186}
(2015) 158-167.

\bibitem {lt:npc}V. E. Levit, D. Tankus, \emph{Complexity Results for Generating Subgraphs}, Algorithmica \textbf{80}
(2018) 2384-2399.

\bibitem {ptv:chordal}E. Prisner, J. Topp and P. D. Vestergaard,
\emph{Well-covered simplicial, chordal and circular arc graphs}, Journal of
Graph Theory \textbf{21} (1996), 113--119.

\bibitem {sknryn:compwc}R. S. Sankaranarayana, L. K. Stewart, \emph{Complexity
results for well-covered graphs}, Networks \textbf{22} (1992), 247--262.

\bibitem {tata:wck13f}D. Tankus, M. Tarsi, \emph{Well-covered claw-free
graphs}, Journal of Combinatorial Theory Ser. B. \textbf{66} (1996) 293-302.

\bibitem {tata:wck13fn}D. Tankus, M. Tarsi, \emph{The structure of
well-covered graphs and the complexity of their recognition problems}, Journal
of Combinatorial Theory Ser. B. \textbf{69} (1997) 230-233.
\end{thebibliography}
\end{document}